\documentclass[11pt,reqno]{amsart}
\usepackage{amsmath,amssymb,graphicx,amscd}
\usepackage{mathrsfs}
\usepackage{enumerate}
\usepackage{mathrsfs}
\usepackage{dsfont}
\usepackage{amsfonts}
\usepackage{fullpage}
\usepackage{color}
\renewcommand{\P}{\mathbf{P}}
\newcommand{\Q}{\mathbf{Q}}
\newcommand{\E}{\mathbf{E}}
\newcommand{\I}{\mathbf{1}}
\newcommand{\FF}{\mathbb{F}}
\newcommand{\GG}{\mathbb{G}}
\newcommand{\F}{\mathcal{F}}
\newcommand{\G}{\mathcal{G}}

\newcommand{\HH}{{\mathbb H}}
\newcommand{\h}{\mathcal{H}}
\newcommand{\RR}{{\mathbb R}}

\newcommand{\finproof}{$\square$}

\newtheorem{thm}{Theorem}[section]
\newtheorem{cor}[thm]{Corollary}
\newtheorem{lem}[thm]{Lemma}
\newtheorem{prop}[thm]{Proposition}
\theoremstyle{definition}
\newtheorem{defn}[thm]{Definition}
\theoremstyle{remark}
\newtheorem*{rem}{Remark}

\theoremstyle{remark}
\newtheorem{ex}[thm]{Example}
\numberwithin{equation}{section}

\numberwithin{equation}{section}

\begin{document}

\title{From the decompositions of a stopping time to risk premium decompositions}

\author{Delia Coculescu}
\address{Swiss Banking Institute \\ University of Z\"urich, Plattenstrasse 32\\ Z\"{u}rich 8032, Switzerland.}
\email{coculescu@isb.uzh.ch}

\begin{abstract}
We build a general model for pricing defaultable claims. In addition to the usual absence of arbitrage assumption, we assume that one defaultable asset (at least) looses value when the default occurs. We prove that under this assumption, in some standard market filtrations, default times are totally inaccessible stopping times; we therefore proceed to a systematic construction of  default times with particular emphasis on totally inaccessible stopping times.

Surprisingly, this abstract mathematical construction, reveals a very specific and useful way in which default models can be built, using both market factors and idiosyncratic factors. We then provide all the relevant characteristics of a default time  (i.e. the Az\'ema supermartingale and its Doob-Meyer decomposition) given the information about these factors.  We also provide explicit formulas for the prices of defaultable claims and analyze the risk premiums that form in the market in anticipation of losses which occur at the default event. The usual reduced-form framework is extended in order to include possible economic shocks, in particular jumps of the recovery process at the default time.  This formulas are not classic and we point out that the knowledge of the default compensator or the intensity process is not anymore a sufficient quantity for finding explicit prices, but we need indeed the Az\'ema supermartingale and its Doob-Meyer decomposition.

\end{abstract}

\maketitle

\section{Introduction}

Negative financial events such as defaults can sometimes be predicted by investors or, on the opposite,  they can occur in an abrupt way and produce losses.  In this paper, the properties of a stopping time which models a default event are  analyzed in relation with the losses that it produces to debt-holders when it occurs, using standard properties of the jump times of martingales. 

In the ''general theory of processes", one classifies stopping times as predictable, accessible and totally inaccessible stopping times (see Definition \ref{classst} below). Traditionally in the default risk literature structural models have produced  default times that  are predictable stopping times (for instance the first hitting time of a fixed level by a diffusion), whereas in the reduced form approach defaults are modeled as totally inaccessible stopping times (first jump of a Cox process). The difference between the two classes of stopping times can be eliminated by a change of the underlying filtration: a predictable stopping time can become totally inaccessible in a smaller filtration, as illustrated by several so-called incomplete information models including \cite{DuffLand}, \cite{Kusuok}. See also the paper \cite{JarrPrott} where this information-based connection between the structural  and reduced-form approaches is explained. 
Using no arbitrage arguments, we show (Section \ref{sec::losses}) that defaults that are thought to produce losses for a financial asset do not have a predictable part  and it is hence natural to model them  as totally inaccessible stopping times in the market filtration.  This gives the economical motivation of our study of properties of totally inaccessible default times.

It has become standard to construct reduced-form default models in two steps (as originally proposed in \cite{elliotjeanbyor}, \cite{jenbrutk1}): one begins with a filtration where the default time is not observable, and then obtain the market filtration after progressively enlarging the original filtration so that the default time becomes a stopping time. In Section \ref{sec::finmodel} we present this construction and some fundamental results that we shall subsequently use in the paper. Our leading assumption in this paper will be that in the enlarged filtration the martingales from the original filtration remain martingales, a property known in the literature as the \textbf{(H)} hypothesis. 

In Section \ref{sec::generalprop} we study some general properties of a  stopping time (i.e., the default time)  under this two-step construction. In particular, we emphasize that there exist decomposing sequences of stopping times in the initial filtration (Proposition \ref{gendecomp}) and we use their properties (i.e., their compensators) in order to characterize the Az\'ema supermartingale of the default time and give its Doob-Meyer decomposition (Proposition \ref{Azema}). This represents a generalization of the classical models, where the Az\'ema supermartingale is supposed to be continuous. The particular case when $\tau$ is  totally inaccessible is also analyzed and we give a useful economic  interpretation to this decomposition.

We also propose a general construction of a random time such that it has a given Az\'ema supermartingale (Section \ref{sec::randomizedst}): given an increasing process $A$ with $A_0=0$ and  $A_\infty=1$,  we construct a sequence of stopping times $(T^i)_{i\in \mathbb{N}}$  and an $\mathbb{N}$-valued random variable $S$ such that the stopping time $T^S$ has its Az\'ema supermartingale equal to the given process $A$. 

Finally, in Section \ref{sec::premium} we apply our results (in particular the decompositions studied in Section \ref{sec::generalprop}) in order to obtain pricing formulas for defaultable claims. The usual reduced-form framework is extended in order to include possible economic shocks, and in particular jumps in the recoveries at the default time.  Indeed, there has been increasing support in the empirical literature that both the probability of default  and the loss given default are correlated and driven by macroeconomic variables. A perfect illustration of this phenomenon is the rapid decline in property prices the recent credit crisis, where defaults coincided with a wave of asset liquidation.  Our aim is to extend the usual reduced-form setup in order to include possible jumps of the recovery process at the default time and to give an expression for the risk premiums attached to such jumps.

Building on the method recently developed in \cite{coculescunikeghbali}, we develop pricing formulas which are not classic (precisely because of the possible jumps of some default-free assets at the default time): the knowledge of the default compensator or the intensity process is not anymore a sufficient quantity for finding explicit prices, but we need indeed the Az\'ema supermartingale and its Doob-Meyer decomposition. We also propose a definition of the default event  risk premium, which measures the compensation investors should require for the losses that occur at the default time (this is strictly positive when the default time is totally inaccessible, but null for a predictable default time) and we compute its expression in a general setting (Theorem \ref{decprem}).

\section{Losses and their predictability}\label{sec::losses}

Let $(\Omega, \F,\HH=(\h_t)_{t\geq0},\P)$ be a filtered probability space.   The filtration $\mathbb H=(\mathcal H_t)_{t\geq 0}$ represents an information flow available to all market investors without cost, that is public information. $\tau$ is an $\RR_+$ valued random variable such that $\tau>0$ a.s. which represents the default time of a company, i.e., the time when the company  is not able to meet some of its financial obligations. Since corporate defaults can be observed when they occur,  $\tau$ is supposed to be an $\HH$ stopping time.

We want to deduce some properties of the particular stopping time $\tau$, given the impact of the default on the prices of the traded assets or portfolio of assets which are available in a financial market. An important question is to know under which economic conditions $\tau$ should be modeled as a totally inaccessible stopping time. Let us recall below a classification of stopping times:

\begin{defn}\label{classst}
[Classification of stopping times] A stopping time $\tau$ is said to be:
\begin{itemize}
\item[(i)] predictable if there exists a sequence of stopping times $(\tau_n)  _{n\geq1}$ such that $\tau_n$ is increasing,
$\tau_{n}<\tau$ on $\{ \tau>0\}  $ for all $n$, and $\lim_{n\rightarrow\infty}\tau_{n}=\tau$ $a.s.$. 

\item[(ii)] accessible if there exists a sequence of\ predictable stopping times $(\tau_{n})_{n\geq1} $, such that:
\[
\mathbb{P}\left(  \cup_{n}\left\{  \omega:\tau\left(  \omega\right)
=\tau _{n}\left(  \omega\right)  <\infty\right\}  \right)  =1.
\]

\item[(iii)] totally inaccessible if, for every predictable stopping time $T$,
\[
\mathbb{P}\left(  \left\{  \omega:\tau\left(  \omega\right)
=T\left( \omega\right)  <\infty\right\}  \right)  =0.
\]
\end{itemize}
\end{defn}

In the line of the now classical asset modeling and pricing, we consider a collection of $\HH$-adapted processes $(S^i_t)_{t\geq 0}$, $i\in\{0,1,...,n\}$  with right continuous with left hand limits sample paths, which represents the evolution of the prices of the assets which are traded in a financial market. Without any loss of generality, we can suppose in this section that the money market account $(S^0_t)$ is constant or, equivalently, that the interest rates are null. We suppose that prices are in equilibrium when regarded as stochastic processes in this filtration, that is, there exist a probability measure $\Q\sim\P$ such that $(S^i_t)$ is a locally bounded local martingale for any $i\in\{1,...,n\}$.  

We assume that there is at least one defaultable asset which is traded on the market (for instance a bond), whose price process is $(S^d_t)$ for some $d\in\{1,...,n\}$. Furthermore we make the following crucial assumption: 
\begin{itemize}\item[\textbf{(L)}] Suppose that $(S^d_t)$ is the price of a defaultable claim. Assume that $\Delta S^d_\tau <0 \; a.s.$, i.e., there is a loss in case of default with probability one.
\end{itemize}
Since the default of a debtor is perceived as a bad news for the creditors, it is reasonable to assume that bond prices or, more generally, prices of defaultable claims, will certainly decrease in the moments after the announcement of the default. Or, in a continuous time, arbitrage free market model, prices decrease with certainty only by negative jumps. Hence, it is reasonable to assume that prices of defaultable claims display negative jumps at the default time. Symmetrically, short positions in defaultable claims would display positive jumps.

We show below that the characteristics of a default time of being predictable or not are in fact intimately related to the validity of the assumption \textbf{(L)}. Therefore, in practical applications, the debate about the properties of a default time can be oriented towards clarifications of the validity of the assumption \textbf{(L)}, which can be tested using market data. 

First let us introduce the restrictions of a stopping time as follows: if $\tau$ is a stopping time and $E\in\F$, then the restriction of $\tau$ to $E$ is $\tau(E):=\tau\I_E+\infty\I_{E^c}$. Notice $\tau(E)$ is a stopping time if and only if $E\in \h_\tau$. We now recall the decomposition of a stopping time. 

\begin{thm}[\cite{cdellacherie}] Let $\tau$ be a $\HH$- finite stopping time. There exists an essentially unique partition of $\Omega$ in two elements $A$ and $B$ of $\h_{\tau-}$ such that the time $\tau(A)$ is accessible and $\tau(B)$ is totally inaccessible. Hence: 
\begin{equation}\label{AB}
\tau=\tau(A)\wedge\tau(B).
\end{equation} $\tau(A)$ is called the accessible part of $\tau$ and $\tau(B)$ is called the totally inaccessible part of $\tau$.
\end{thm}

\begin{rem}
Let us point out that the decomposition (\ref{AB}) is stable under equivalent changes of the probability measure. 
\end{rem}

\begin{prop} Under  \textbf{(L)}, the default time $\tau$ does not have a predictable part, that is,  there does not exist $E\in\h_\tau$ such that $\tau(E)$ is a finite predictable stopping time. In particular, when it is finite, the accessible part of $\tau$ is not predictable. \end{prop}

\proof It is equivalent to show the result under an equivalent  local martingale measure for $S^d$, i.e., such that the process $(S^d_t)$ is a local martingale. Therefore we assume without loss of generality that $\P$ is a local martingale measure. Denote $\delta:=\Delta S^d_\tau$. By \textbf{(L)}, $\delta <0$ $a.s.$. 

Suppose that there exists a set $E\in \h_\tau$ such that $\tau(E)$ is a predictable stopping time. Then we need to show that $\P(E)=0$ hence $\tau(E)=\infty$  $a.s.$. 

Let $(T^n)_{n\in \mathbb{N}}$ be a reducing sequence of stopping times for $(S^d_t)$, that is $(T^n)$ increases to $\infty$ and the stopped process $S^n_t:=S^d_{t\wedge T^n}$ is a uniformly integrable martingale. The predictable stopping theorem states that:
\[
\E[\Delta S^n_{\tau(E)}|\h_{\tau(E)-}]=0 \text { and implies that } \E\left[\Delta S^n_{\tau(E)}\right]=0.
\]
Notice that $\Delta S^n_{\tau(E)}=\delta\I_{E}\I_{\{\tau\leq T^n\}}$ therefore we must have:
\[
\E[\delta\I_{E\cap\{\tau\leq T^n\}}]=0.
\]
Since $\delta$ is strictly negative it follows that 
\begin{equation}\label{pn}
\P(E\cap\{\tau\leq T^n\})=0 \text{ for all $n\in \mathbb{N}$. }
\end{equation}
Notice that $\{\tau\leq T^n\}\subset \{\tau\leq T^{n+1}\}$ hence taking the limit in (\ref{pn}) as $n\to \infty$ leads to $P(E)=0$.
\finproof

The above proposition reveals that the structural models of default are not compatible with the assumption  \textbf{(L)}, since the default times in these models are modeled as hitting times of some observable (i.e., $\HH$ adapted) diffusions, or jump diffusions, hence have a predictable part. A way to obtain the condition  \textbf{(L)} from a structural model is to assume that the process triggering the default is not $\HH$ adapted or in other words is not observable by common market investors (see for instance the incomplete information models in \cite{DuffLand}, \cite{GieseGold1}, \cite{cetjarproy},  \cite{GuoJarrZeng}, \cite{cocgemjea08}, \cite{FreySchmidt}). A general discussion about the links between structural, reduced-form and imperfect information models is provided in \cite{JarrPrott}.

It seems that totally inaccessible stopping times are suitable to be used when one is interested to model situations characterized in the assumption  \textbf{(L)}. Let us point out that under \textbf{(L)} the default time can nevertheless have an accessible part in some types of filtrations, since there exist in general accessible stopping times which are not predictable. 

\begin{cor} Suppose that the filtration $\HH$ is quasi left-continuous i.e., $\h_S=\h_{S-}$ for all predictable stopping times $S$. Under  \textbf{(L)}, the default time $\tau$ is a totaly inacessible stopping time.
\end{cor}
\begin{proof}
If the filtration $\HH$ is quasi left-continuous,  then all accessible stopping times are predictable (T51, page 62, \cite{cdellacherie}).
\end{proof}

Quasi left continuity of a filtration is not a very intuitive notion form an economical point of view. Let us point  out that $\HH$ is quasi left-continuous for instance when the filtration is generated by a L\'evy process, this being the class of models the most commonly used in modeling security prices. Therefore, the above result suggest that different filtrations (i.e., which are not quasi left-continuous) should be used if one wants to capture some specific features of debt, such that possibilities to default at some fixed dates in time together with the assumption \textbf{(L)}. In this case, accessible but not predictable stopping times should be used. We provide below an example of such a construction:
\begin{ex} 
Suppose that a bond is expected to pay coupons at the fixed dates $t^i, i=1,...,n$ and the bond can default only at one of the coupon dates. We begin with a filtration $\FF$ and a sequence $(b^i)_{i\geq 1}$ of Bernoulli random variables independent from each other and independent from the filtration $\FF$. We denote $i^*:=\inf\{i|b^i=1\}$,  the default time: $\tau=t^{i^*}\I_{\{i^*\leq n\}}+\infty \I_{\{i^*>n\}}$ and the market filtration: $\h_t:=\F_t\vee \sigma(\tau\wedge t)$ for  $t\geq 0$. It is easy to check that  the default time $\tau$  is an accessible but not predictable $\HH$- stopping time, so that the assumption \textbf{(L)} can be implemented in this setting.
\end{ex}

\section{The two-step construction of the information set}\label{sec::finmodel}

In this short section we fix the mathematical framework that will be used for the rest of the paper. Following \cite{bsjeanblanc}, \cite{elliotjeanbyor}, \cite{jenbrutk1}, we shall construct a default model in two steps. The main idea is to separate those assets which do not default at time $\tau$, form those who are defaultable at $\tau$. Defaultable assets are those issued by the particular company we are analyzing, while the remaining traded assets (that, to simplify, we call default-free) are issued by other companies or by  governments, they may as well be indexes, commodities or derivatives: they should all have different default times than $\tau$.  
Let $(S^i)_{i\in I}$, with $I\subset \{1,...,n\}$ represent the price processes of the default-free assets and $\FF$ their natural $\P$ augmented filtration. The default model is built using the information $\FF$ about the default-free assets (instead of all the public information $\HH$) and the default arrival as follows.

Let $(\Omega,\F,\FF=(\F_t)_{t\geq 0},\P)$ be a filtered probability space satisfying the usual assumptions. The default time $\tau$ is defined as a random time (i.e., a nonnegative $\F$-measurable random variable) which is not an $\FF$-stopping time.  We assume throughout that $\P(\tau=\infty)=0$.

Then, a second filtration $\GG=(\G_t)_{t\geq 0}$ is obtained by progressively enlarging the filtration $\FF$ with the random time $\tau$: $\GG$ is the smallest filtration satisfying the usual assumptions, containing the original filtration $\FF$, and for which $\tau$ is a stopping time, such as explained in \cite{jeulin}, \cite{jeulinyor}:
$$\G_{t}=\mathcal{K}_{t+}^{o} \quad \text{where} \quad \mathcal{K}_{t}^{o}=\F_{t}\vee\sigma(\tau\wedge t).$$

The advantage of this construction  is that it relies on projections of some $\GG$ adapted processes onto the filtration $\FF$ and therefore displays some finer properties of the default time, which are very useful for pricing and hedging. 

\textit{Notation.} If $X$ is a measurable process, we  denote by $\:^oX$ (resp. $\:^pX$) its  $\FF$ optional (resp. $\FF$ predictable) projection and if $X$ is increasing we denote by $X^o$ (resp. $X^p)$ its  $\FF$ dual optional (resp. $\FF$ dual predictable) projection. The definitions of these notions can be found in the Appendix.

The following processes shall play a crucial role in our discussion:
\begin{itemize}
\item the $\FF$ supermartingale
\begin{equation}
Z_t^\tau=\mathbf{P}\left[ \tau >t\mid \mathcal{F}_{t}\right]=\;^o(1_{\{ \tau>\cdot\} })_t
\label{surmart}
\end{equation}
chosen to be c\`{a}dl\`{a}g, associated with $\tau $\ by Az\'{e}ma
(\cite{azema}), (note that $Z_t>0$ on the set $\{t<\tau\}$); 

\item the $\FF$ dual projections $A_{t}^{\tau }:=(\I_{\{ \tau \leq
\cdot\}})^o_t$ and $a_{t}^{\tau }=(\I_{\left\{ \tau \leq \cdot \right\}})^p_t$;

\item the c\`{a}dl\`{a}g martingale
\begin{equation*}
\mu _{t}^{\tau }=\mathbf{E}\left[ A_{\infty }^{\tau }\mid \mathcal{F}_{t}\right] =A_{t}^{\tau }+Z_{t}^{\tau }.
\end{equation*}

\item the Doob-Meyer decomposition of (\ref{surmart}):
\begin{equation*}
Z_t^\tau=m_{t}^{\tau }-a_{t}^{\tau },
\end{equation*} where $m^{\tau}$ is an $\FF$-martingale.
\end{itemize}
A common assumption in the literature is that the random time $\tau$ avoids the $\FF$ stopping times, that is $\P(\tau=T)=0$ for any $\FF$ stopping time $T$. In this case, $A^\tau=a^\tau$ is continuous, a property useful later on. 

The Az\'ema supermartingale $(Z^\tau_t)$  is the main tool for computing  the $\GG$ predictable compensator of $\I_{\{\tau\leq t\}}$:
\begin{thm}[\cite{yorjeulin}]\label{calccomp}
The process:
$$N_t:=\I_{\{\tau\leq t\}}-\int_{0}^{t\wedge\tau}\frac{1}{Z_{s-}^{\tau}}da_{s}^{\tau}$$
is a $\GG$ martingale.
\end{thm}
We shall assume in the rest of the paper that:
\begin{itemize}
\item[\textbf{ (H)}] Every $\FF$ (local-)martingale is a $\GG$ (local-)martingale.
We say that the filtration $\FF$ is immersed in $\GG$ and denote this property by the symbol: $\FF\hookrightarrow \GG$.
\end{itemize}
The immersion property was studied in \cite{bremaudyor}, \cite{delmeyfil}. This assumption can be related to absence of arbitrages in the market model (see \cite{bsjeanblanc}, \cite{jealec08}, \cite{cocjeanik09}). 

It is known that in our framework  \textbf{(H)} is equivalent to the following: For all $s\leq t$,
\begin{equation}\label{H}
\mathbf{P}\left[ \tau \leq s\mid \mathcal{F}_{t}\right]
=\mathbf{P}\left[ \tau \leq s\mid \mathcal{F}_{\infty }\right] .
\end{equation}
It follows that under  \textbf{(H)} , the Az\'ema supermartingae is a decreasing process, i.e., 
\begin{equation}\label{Zdecr}
Z^\tau=1-A^\tau.
\end{equation} 

Moreover, when the immersion property holds, simple projection formulas for stochastic integrals hold \cite{bremaudyor}:

\begin{prop}[\cite{bremaudyor}]\label{bremyorproj} Suppose that $\FF\hookrightarrow \GG$.
\begin{itemize}
\item[(i)] Let $M$ be an $\FF$ local martingale and $H$ a bounded process. Then:
\begin{equation*}
\;^o\left(\int HdM\right)=\int \;^o\left(H\right)dM.
\end{equation*}
\item[(ii)] If $M$ is a $\GG$ square integrable martingale and $H$ an $\FF$ bounded process. Then:
\begin{equation*}
\;^o\left(\int HdM\right)=\int Hd\;^o(M).
\end{equation*}
\end{itemize}
\end{prop}

\section{Decompositions of a stopping time}\label{sec::generalprop}

We now investigate some general properties of a stopping time $\tau$ of a filtration $\GG$ constructed by progressively enlarging a filtration $\FF$, as explained in Section \ref{sec::finmodel}.   We begin with a financial example and then show a general construction. 

Suppose that $T^k, k\in N^* \subset \mathbb{N}^*$ are unexpected times when negative shocks occur in the economy. They are finite, totally inaccessible stopping times in the filtration $\FF$. Morover $\P(T^i=T^j)=0, i\neq j$. Now, we suppose that a firm's financial health is affected by these shocks, and this impact may be so severe that can trigger the default of the company i.e.:
\[
\P(\tau=T^k)>0, \forall k\in N^*.
\]
Denote: $p^k_t:=\P(\tau=T^k|\F_t)$ such that we obtain: $$\sum_{k\in N^*} p^k_t\leq 1.$$

\begin{lem}
\begin{itemize}

\item[(a)] If $\sum_{k\in N^*} p^k_0 =1$ then $\tau$ is a $\GG$ totally inaccessible stopping time.

\item[(b)] If $\tau$ is a $\GG$ totally inaccessible stopping time and $\sum_{k\in N^*} p^k_0 <1$, then  there exists a $\GG$ totally inaccessible stopping time $T^0$ such that:
\begin{equation}\label{taudecomp1}
\tau=\sum_{i\in N} T^i \I_{\{T^i=\tau\}}.
\end{equation}
 where $N=N^* \cup \{0\}$ and hence $\sum_{k\in N} p^k_t=1$, where $p^0_t=\P(\tau=T^0|\F_t)$.

\end{itemize}
\end{lem}
\proof
(a) Follows from the definition of a totally inaccessible stopping time. 

(b) Denote:

\begin{align}\label{tzero}
T^0:& =\tau\I_{\{\tau \neq T^i, \forall i\in N^*\}}+\infty \I_{\{\exists i\in N^*, \tau=T^i\}}\\
&=\tau \Pi_{i\in N^*}\I_{\{\tau \neq T^i\}}+\infty \sum_{i\in N^*} \I_{\{\tau = T^i\}} \nonumber
\end{align}

Obviously $T^0$ is a $\GG$ stopping time since $\{T^0 <t\}=\{\tau<t\} \bigcap_{i\in N^*} \{T^i>\tau\}$ is in $\G_t$. Let $S$ be any predictable $\GG$ stopping time. Then, $\P(T^0=S<\infty)=\P(\tau\I_{\{\tau \neq T^i, \forall i\in N^*\}}=S<\infty)\leq \P(\tau=S<\infty)=0$ since $\tau$ is totally inaccessible.
\finproof

\begin{rem}
Without loss of generality, we may assume that $T^0$ avoids all $\FF$ finite stopping times, that is:
\[
\P(T^0=T <\infty) =0 \text{ for any $\FF$ stopping time } T.
\] Indeed, since it is totally inaccessible, $T^0$ avoids the finite predictable and accessible $\FF$ stopping times. If there exists an $\FF$ stopping time $T^{s}$  such that $\P(T^0=T^{s}<\infty)>0$, then, from \ref{taudecomp1} $s \notin N^*$.  We denote $\bar T^0= T^0\I_{\{T^0\neq T^{s}\}}+\infty \I_{\{T^0 = T^{s}\}}$. Then the times $T^i, i\in N^*\cup \{s\}$ are the times of economic shocks and $\bar T^0$ is $\GG$ totally inaccessible which avoids all finite $\FF$ stopping times. Also $\tau=\sum_{i\in N^*\cup \{s\}} T^i \I_{\{T^i=\tau\}}+\bar T^0 \I_{\{\bar T^0=\tau\}}.$
\end{rem}
We now show that it is natural to construct default models as above, since any stopping time $\tau$ admits decompositions involving sequences of $\FF$ stopping times as follows: 

\begin{prop}\label{gendecomp} Let $\tau$ be a finite $\GG$ stopping time. Then, there exists a sequence $(T^i)_{i\geq 1}$, of   $\FF$ stopping times, such that 
\begin{equation}\label{prop1}
\P(T^i=T^j <\infty) =0 \quad i\neq j
\end{equation} and 
\begin{equation}\label{prop2}
\P(\tau=T^i)>0\text{ whenever } \P(T^i<\infty)>0
\end{equation}
and  a totally inaccessible $\GG$ stopping time $T^0$ such that $T^0$ avoids all finite $\FF$  stopping times and such that:
\begin{equation}\label{taudecomp}
\tau=\sum_{i\geq 0} T^i \I_{\{T^i=\tau\}}.
\end{equation}

The $\GG$ stopping time $\tau$ is totally inaccessible if and only if the $\FF$ stopping times $(T^i)_{i\geq 1}$ are totally inaccessible.
\end{prop}

\begin{proof} We begin with a useful lemma:
\begin{lem}\label{ja}
Let  $T$ be an $\FF$ stopping time.  $\P(T=\tau)>0$ if and only if the event $\{A^\tau_T(\omega)\neq A^\tau_{T-}(\omega)\}$ has a strictly positive probability (one says that  $A^\tau$ charges $T$).
\end{lem}
\begin{proof}
Let us recall that $A^\tau$ is the dual optional projection of the increasing process $H_t:=\I_{\{\tau\leq t\}}$, which has a unique jump of size $1$ at the time $\tau$.  Therefore, for any stopping time $T$ we have: $\Delta H_T=\I_{\{\tau=T\}}$. 
If $T$ is an $\FF$ stopping time, then by Theorem VI.76. in \cite{delmeyV}:
\begin{equation*}
\Delta A^\tau_T=\E[\Delta H_T|\F_T]=\P(\tau=T|\F_T).
\end{equation*} Recall that $A^\tau$ is increasing, hence $\Delta A^\tau$ is nonnegative. The result follows.
\end{proof}

We now prove the proposition. There are two possible situations:
\begin{itemize}
\item[1.]  $A^\tau$ continuous. This corresponds to the situation when $\tau$ avoids all $\FF$ stopping times. Therefore the decomposition (\ref{taudecomp}) holds with all $(T^i)_{i\geq 1}$ infinite a.s. and $T^0=\tau$ $a.s.$.

\item[2.] $A^\tau$ discontinuous.  Let $T^i$, $i\geq1$ be the ordered jump times of $A^\tau$. If there are finitely many, we simply set the remaining stopping times in the sequence to be infinite (notice that (\ref{prop1}) is satisfied by this sequence). We obtain from Lemma \ref{ja}  that:
\begin{itemize}
\item[(i)] $\P(\tau=T^i)>0$ whenever $\P(T^i<\infty)>0$;
\item[(ii)] $\tau$ avoids any $\FF$ stopping time with the graph disjoint of the union of the graphs of $T^i$, $i\geq 1$.
\end{itemize}
We define $T^0$ as the restriction of $\tau$ to the set $E=\left(\cup_{i\geq 1}\{T^i=\tau\}\right)^c$. $T^0$ is a $\GG$ stopping time since $E \in \G_\tau$. By construction, $T^0$ satisfies, 
\begin{equation}\label{t0}
\P(T^0=T^i<\infty)=0\quad \forall i\geq 1.
\end{equation}
From this and $(ii)$ it follows that $T^0$ avoids all finite $\FF$-stopping times. Notice that if $\P(E)=0$ then $ T^0=\infty$ avoids the finite $\FF$ stopping times and is totally inaccessible ($\infty$ is the only time which is both accessible and totally inaccessible). 
\end{itemize}The last statement in the proposition follows from the definition of a totally inaccessible stopping time.
\end{proof}

\begin{rem} From the above proof, we see that we can in fact choose $(T^i)_{i\geq1}$ to be any sequence that exhausts the jump times of the process $(A^\tau_t)$ and therefore the sequence is not uniquely defined. Useful later on will be to notice that the sequence $(T^i)_{i\geq1}$ can be chosen such that it contains only totally inaccessible and predictable $\FF$ stopping times (indeed, one can decompose a stopping time first in its accessible and totally inaccessible part, and then the accessible part in a sequence of predictable times). On the other hand, the time $T^0$ can be uniquely defined (up to  null sets) if chosen to be infinite on the set $\{\tau=T^i, i\geq 1\}$ ( as the construction in equation  (\ref{tzero})  with $N^*=\mathbb{N}^*$).
\end{rem}

Now, we provide an economic interpretation of the sequence $(T^i)$, when they are totally inaccessible stopping times. As the filtration $\FF$ is generated by the prices of default-free claims, there exist some default free assets which have jumps at the $\FF$ stopping times $(T^i)_{i\in \mathbb{N}^*}$ which are not infinite. Therefore, when the default arrives at one of these times, some default-free asset prices  react abruptly by jumps. The interpretation is that the default $\tau$ has a macroeconomic impact (or is triggered by some macroeconomic shock). On the opposite, since $T^0$ avoids all finite $\FF$-stopping times, when default arrives at this time (i.e., on $\{\tau=T^0\}$, there will be no impact of the default event on the prices of the default-free assets, that is no default-free asset price will jump. Therefore, we propose the following:

 \begin{defn}
Let $\tau$ be a totally inaccessible default time. Take a decomposition of $\tau$ as in (\ref{taudecomp}). Then, we shall call $T^i, i\in \mathbb{N}^*$ times of macroeconomic shocks, and $T^0$ is the idiosincratic default time (i.e. when default is due to the unique and specific circumstances of the company, as opposed to the overall market circumstances). 
\end{defn}

In the remaining of this section, we shall derive the useful properties of $\tau$, given properties of a decomposing sequence of  times $(T^i)_{i\in \mathbb{N}}$. $\tau$ and the sequence $(T^i)_{i\in \mathbb{N}}$ will be always supposed to fulfill the properties stated in Proposition \ref{gendecomp} (except $T^0$, they are not necessarily totally inaccessible). Moreover, we suppose $T^0$ to be infinite on the set $\{\tau\neq T^0\}$. 

Let us denote by $(\Lambda^i_{T^i \wedge t})$ the $\FF$-compensators of the $\FF$ stopping times $T^i$, $i\in \mathbb{N}^*$. It follows that $(\Lambda^i_{t\wedge T^i})$, $i\in \mathbb{N}^*$ are $\FF$ adapted and since $\FF \hookrightarrow \GG$ they are also the $\GG$ compensators of $T^i$, $i\in \mathbb{N}^*$.  We introduce the $\FF$ (and $\GG$)-martingales: 
\[
N^i_t:=\I_{\{T^i\leq t\}}-\Lambda_{t\wedge T^i}, \text{ for } i\in \mathbb{N}^*.
\]

On the other hand, the time $T^0$ is not an $\FF$ stopping time. We denote  $a^0_t$ the $\FF$ dual optional projection of $\I_{\{T^0\leq \cdot\}}$. Since $T^0$ avoids all $\FF$ stopping time, $a^0$ is continuous. Let us introduce $\tilde\F_t:=\F_t\vee\sigma(t\wedge T^0)$ and notice that $\tilde\F_t\subset\G_t$, $t\geq 0$.  Since $\FF\hookrightarrow \GG$ it follows that $\FF\hookrightarrow \tilde \FF\subset\GG$, hence the  Az\'ema supermartingale of the time $T^0$ is decreasing (see equation (\ref{Zdecr})) and equals $$Z^0_t:=\P(T^0>t|\F_t)=1-a^0_t.$$ Notice that $Z^0_{\infty}=\P(T^0=\infty|\F_\infty)=1-a^0_\infty.$  Let us also denote $(\Lambda^0_{T^0 \wedge t})$ the $\GG$-compensator of $T^0$, such that: 
\[
N^0_t:=\I_{\{T^0\leq t\}}-\Lambda_{t\wedge T^0}, t\geq 0
\]
is a $\GG$-martingale. 

Now, we give the decomposition of the compensator of $\tau$ given those of the times $(T^i)$. 
\begin{prop}\label{propcompensator} Let $g^i_t: =\P(\tau=T^i|\G_t)$ for $t\geq 0$ and  $i\in \mathbb{N}$.  The predictable compensator of $\I_{\{\tau\leq \cdot\}}$ denoted $(\Lambda_{t\wedge\tau})$ satisfies: 
\begin{equation*}\label{lambda}
\Lambda_{t\wedge\tau}=\sum_{i\geq 1}\int_0^{t}(g^i_{s-}+u^i_s)d\Lambda^i_{s\wedge T^i}+\Lambda^0_{T^0\wedge t},
\end{equation*}
where for $i\geq 1$, $(u^i_t)$ is a predictable process that satisfies  $\langle g^i,N^i\rangle_t=\int_0^t u^i_sd\Lambda^i_{T^i\wedge s}$. 
Hence, the martingale $N_t=\I_{\{\tau\leq t\}}- \Lambda_{t\wedge\tau}$ decomposes as:
\[
N_t=\sum_{i\geq 1}\left(\int_0^{t}(g^i_{s-}+u^i_s)dN^i_s+(\Delta g^i_{T^i}-u^i_{T^i})\I_{\{T^i\leq t\}}\right)+N^0_t.
\]
Therefore $\tau$ has a $\GG$ intensity $\lambda$, i.e., $\Lambda_{t\wedge \tau}=\int_0^{t\wedge \tau}\lambda_s ds$  if and only if $\GG$ intensities exist for the times $T^i$, $i\geq 0$. Then, denoting by $\lambda^i$ the $\GG$ intensity of $T^i$, the following relation holds:
\begin{equation}\label{intensity}
\lambda_t=\sum_{i=1}^n  (g^i_{t}+u^i_s)\lambda^i_{t}\I_{\{T^i>t\}}+\lambda^0_t.
\end{equation}
\end{prop}

\proof The result follows from:
\begin{equation}\label{decI}
\I_{\{\tau\leq t\}}=\sum_{i\geq 0}\I_{\{T^i=\tau\}}\I_{\{T^i\leq t\}}=\sum_{i\geq 1} g^i_{T^i}\I_{\{T^i\leq t\}}+\I_{\{T^0\leq t\}}
\end{equation}
since $\{T^0\leq t\}\subset \{\tau=T^0\}$. By dominated convergence $\;^p(\sum_{i\geq 1}g^i_{T^i}\I_{\{T^i\leq t\}})=\sum_{i\geq 1}\;^p(g^i_{T^i}\I_{\{T^i\leq t\}})$. We notice that:
\[g^i_{T^i}\I_{\{T^i\leq t\}}=\int_0^tg^i_{s-}d\I_{\{T^i\leq s\}}+\Delta g^i_{T^i}\I_{\{T^i\leq t\}}
\]
For $i\geq 1$, there exists a  $\GG$-predictable process $(u^i_t)$ such that  $\langle g^i,N^i\rangle_t=\int_0^t u^i_sd\Lambda^i_{T^i\wedge s}$ (see \cite{delmey78}). We have that $\;^p\left(\int g^i_{s-}d\I_{\{T^i\leq \cdot\}}\right)=\int g^i_{s-}d\Lambda^i_{T^i\wedge\cdot}$, therefore we only need to show the equality: $\;^p\left(\Delta g^i_{T^i}\I_{\{T^i\leq t\}}\right)=\langle g^i,N^i\rangle_t$. As we have remarked previously (see the Remark after Proposition \ref{gendecomp}) it is sufficient to consider the cases when $T^i$ is either totally inaccessible or predictable. If $T^i$ is totally inaccessible, then $\Delta g^i_{T^i}\I_{\{T^i\leq t\}}= [g^i,N^i]_t$ since $N^i$ is of finite variation and on $\{T^i\leq t\}$ has one jump of size $1$ at $T^i$, so that the equality is correct. If $T^i$ is predictable, then $\;^p\left(\Delta g^i_{T^i}\I_{\{T^i\leq t\}}\right)=0$ by the predictable stopping theorem. On the other hand, since $N^i\equiv 0$, $\langle g^i,N^i\rangle_t=0$ and again the needed equality is correct.
\finproof

We now have a short lemma:
\begin{lem}\label{pis}
If $\FF\hookrightarrow \GG$, then for $i\in \mathbb{N}^*$, the martingales $p^i_t:=\P(\tau=T^i|\F_t)$ satisfy $p^i_{t}=p^i_{t\wedge T^i}$.
\end{lem}
\proof
Let us point out that if $\FF\hookrightarrow \GG$ then for any $\FF$ stopping time $T$ and and $\G_T$ measurable random variable $G$ we have:
\[
\E[G|\F_\infty]=\E[G|\F_T].
\]
Then, the result follows easily from the fact that $p^i$ is the optional projection of the martingale $g^i$ which also satisfies  $g^i_{t}=g^i_{t\wedge T^i}$ and from Proposition \ref{bremyorproj}:
\[
p^i_\infty=\E[g^i_{\infty}|\F_{\infty}]=\E[g^i_{T^i}|\F_{\infty}]=\E[g^i_{T^i}|\F_{T^i}]=p^i_{T^i}.
\]
\finproof

The next proposition  contains the important properties of  $\tau$ seen as an $\FF$ random time.  

\begin{prop}\label{Azema} Suppose that $\FF\hookrightarrow \GG$. The Az\'ema supermartingale of $\tau$ is given by:
\begin{equation}\label{Z}
Z^\tau_t:=1-\left(\sum_{i\geq 1} p^i_{T^i}\I_{\{T^i\leq t\}}+a^0_t\right),
\end{equation}
and its Doob-Meyer decomposition is:
\begin{equation*}
Z^\tau_t=\left(1-\hat N_t \right)-a^\tau_t,
\end{equation*}
where:
\[
a^\tau_t:=\sum_{i\geq1}\int_0^t (p^i_{s-}+\upsilon^i_s) d\Lambda^i_{T^i\wedge s}+a^0_s;
\]
\begin{equation}\label{Nhat}
\hat N_t: =\E[N_t|\F_t]=\sum_{i\geq 1} \int_0^t (p^i_{s-}+\upsilon_{s})dN^i_{s}+(\Delta p^i_{T^i}-\upsilon^i_{T^i})\I_{\{T^i\leq t\}}
\end{equation}
and for $i\geq 1$, $\upsilon^i$ is an $\FF$ predictable processes that satisfies  $\langle N^i,p^i\rangle =\int \upsilon^i_sd\Lambda^i_{s\wedge T^i}$.
 \end{prop}

\proof  $A^\tau_t=\sum_{i\geq 1} p^i_{T^i}\I_{\{T^i\leq t\}}+a^0_t$ is easily computed as the optional projection of $\I_{\{\tau\leq \cdot\}}$ expressed as in (\ref{decI}). Hence we obtain the formula (\ref{Z}) from  $Z^\tau=1-A^\tau$.

By definition $a^\tau$ is the $\FF$ compensator of the process $A^\tau$. Its expression can be found using exactly the same arguments that the ones in the proof of the Proposition \ref{propcompensator}, with $p^i$ instead of $g^i$ and $\upsilon^i$ instead of $u^i$. Therefore, by definition of $A^\tau$ and $a^\tau$, we have that $\hat N=A^\tau-a^\tau=\;^o(N)$ is a martingale.  \finproof

Notice that when the times $(T^i)_{i\geq 1}$ are totally inaccessible, then the processes $\Lambda^i$ and $a^0$ are continuous, the process $a^\tau$ is indeed continuous as required by the property of $\tau$ being totally inaccessible.
Also, the above Proposition makes it clear that when $\tau$ does not avoid all $\FF$ stopping times $Z^\tau$ is discontinuous and  $A^\tau$ and $a^\tau$ may differ. More exactly:
\begin{cor} The dual optional and predictable projections of $\tau$ coincide, that is $A^\tau=a^\tau$ if and only if $\tau$ avoids all totally inaccessible stopping times.
\end{cor}

Also, using (\ref{calccomp}) we can write the intensity $\lambda$ in (\ref{intensity}) using $\FF$ adapted processes, as: 
\[
\lambda_t=\sum_{i=1}^n  \left(\frac{p^i_{t-}+\upsilon^i_s}{Z^\tau_s}\right )\lambda^i_{t}\I_{\{T^i\geq t\}}+\lambda^0_t.
\]
Notice that when $\tau$ has an intensity (as it is classical in the reduced-form approach) the times  $T^i$, $i\geq 1$ can also be identified as times of jumps in the intensity of $\tau$, which is more immediate than seeing them as times of jumps of the process $A^\tau$ which is often not modeled explicitly.

\section{A general construction}\label{sec::randomizedst}

Starting with an increasing  c\`adl\`ag process $(A_t)_{t\geq 0}$ with  $A_0=0$ $a.s.$ we construct of a family of random times  $(T^i)_{i\in \mathbb{N}}$  and an $\mathbb{N}$-valued random variable $S$ such that the stopping time $T^S$ has its Az\'ema supermartingale equal to the given process $A$ (Theorem \ref{genconstr}). 

\begin{ex} Let $(T^i)_{i\in E}$, $E\subset  \mathbb{N}$  be a family of $\FF$ stopping times and 
 $S$ an $E$ valued random variable independent from $\F_\infty$. Let the default time be:
\[
\tau=T^S. 
\]
Then, $\tau$ is a $\GG$ stopping time such that $\FF \hookrightarrow \GG$. 
\end{ex}

The above example can be generalized as follows:

\begin{thm}\label{thmrand1}
Let $(T^i)_{i\in E}$, $E\subset \mathbb{N}$ be a family of $\FF$ stopping times and $S$ an $E$ valued random variable with $p^i_t:=\P(S=i|\F_t)$.  We introduce the random time:
\[
\tau=T^S. 
\]
Then, $\FF \hookrightarrow \GG$ if and only if  $p^i_t=p^i_{t\wedge T^i}$, for all $i \in E$. In this case:
\begin{equation}\label{randomZ}
\P(\tau\leq t|\F_t):=\sum_{i\geq1}  p^i_{T^i}\I_{\{T^i\leq t\}}
\end{equation}
\end{thm}
\proof
If $\FF \hookrightarrow \GG$  then the properties $p^i_t=p^i_{t\wedge T^i}$, for all $i \in E$ follow by the same arguments as in the proof of Lemma \ref{pis}. Now we show that this condition is sufficient for $\FF \hookrightarrow \GG$  to hold.

It is known that $\FF \hookrightarrow \GG$ if and only if for $s\geq t$, $\P(\tau\leq t|\F_s)=\P(\tau\leq t|\F_t)$ (see (\ref{H})).

Let us denote $E^i=\{S=i\}\bigcap\{T^i\leq t\}$ and notice that $E^i$ and $E^j$ are disjoint events for $i\neq j$. We obtain for $s\geq t$ (we use the Monotone Convergence Theorem for interchanging summation and expectation):
\begin{equation*}
\P(\tau\leq t|\F_s)= \P\left(\bigcup_{i\geq 0}E^i|\F_s\right)=\sum_{i\geq 0} \P\left(E^i|\F_s\right) =\sum_{i\geq 1}\I_{\{T^i\leq t\}} p^i_s.
\end{equation*} If  $p^i_t=p^i_{t\wedge T^i}, t\geq 0$ then expression (\ref{randomZ}) follows as well as the equality $\P(\tau\leq t|\F_s)=\P(\tau\leq t|\F_t)$.
\finproof

It is well known that is possible to associate with a given increasing process a random time (see for instance \cite{elliotjeanbyor}). Suppose that our probability space supports a random variable $\Theta$ uniform on $[0,1]$ which is independent from the sigma field $\F_\infty$. Assume we are given an $\FF$ adapted, increasing c\`adl\`ag process $(A_t)_{t\geq 0}$ with  $A_0=0$ $a.s.$ and $A_\infty=1$ $a.s.$. Then, 
\begin{equation}\label{constr1}
\tau:=\inf\{t | A_t\geq \Theta\}.
\end{equation}
satisfies $\P(\tau\leq t|\F_t)=A_t.$

Below we show an alternative construction of $\tau$  which emphasizes the role of sequences of $\FF$ stopping times.

\begin{thm}\label{genconstr} Assume we are given an $\FF$ adapted, increasing c\`adl\`ag process $(A_t)_{t\geq 0}$ with  $A_0=0$ $a.s.$ and $A_\infty=1$ $a.s.$.
Let $(T^i)_{i\geq 1}$ be any sequence of stopping times exhausting the jumps of $A$ such that $\P(T^i=T^j)=0$, for $i\neq j$. 
Denote:
\begin{align*}
p^i_t&:=\E[\Delta A_{T^i}|\F_t]\quad \forall i\geq 1,\\
p^0_t&:=\E[A^c_\infty|\F_t]
\end{align*} where $(A^c_t)$ is the continuous part of $(A_t)$.
Let us suppose that our probability space supports a random variable $\Theta$ uniform on $[0,1]$, independent from the sigma field $\F_\infty$. Denote 
\[a_t:=1-\exp\left\{-\int_0^t\frac{dA^c_s}{p^0_s-A^c_s}\right\}
\]
and:\[
T^0=\inf\{t | a_t>\Theta\}.
\]

Suppose furthermore that our probability space supports a second random variable $S:\Omega \to \mathbb{N}$ which satisfies:
\begin{equation}\label{condlawS}
\P(S=i|\F_\infty\vee\sigma(\theta))=p^i_0+\int_0^{T^0}\frac{dp^i_s}{1-a^0_s}.
\end{equation}

Then, the random time
\[
\tau=T^S. 
\]
satisfies  $\P(\tau\leq t|\F_t)=A_t.$ Furthermore, $\FF \hookrightarrow \GG$.
\end{thm}

\proof First let us have a look to the integral appearing in the definition of the process $(a_t)$. This is an increasing process which converges to infinity. Its explosion time is $\nu=\inf\{t\;|\;A^c_t=p^0_t\}$ which is an $\FF$ stopping time. On the stochastic interval $[\nu,\infty)$ the martingale $p^0_t$ and $A^c_t$ stay constant. We can therefore also write:
\[a_t:=1-\exp\left\{-\int_0^{t\wedge\nu}\frac{dA^c_s}{p^0_s-A^c_s}\right\}.\] 
Remark that $(a_t)$ is a continuous increasing process with $a^0=0$ and $a_\infty=1$.

Let us now introduce the following filtration $\GG^0=(\G^0_t)_{t\geq 0}$:
\[
\G^0_t :=\F_{t}\vee\sigma(T^0 \wedge t). 
\]
It is easy to check that $\FF \hookrightarrow \GG^0$ by the property (\ref{H}) since $\Theta$ is independent from $\F_\infty$ hence:
\begin{equation}\label{a}
\P(T^0\leq t|\F_t)=\P(T^0\leq t|\F_\infty)=a_t.
\end{equation}
 Therefore the $\FF$ martingales  $(p^i_t)$ are also $\GG^0$ martingales. Since $\G^0_\infty=\F_\infty\vee \sigma(\theta)$, we obtain that:
\[q^i_t:=\P(S=i|\G^0_t)=\E[\P(S=i|\G^0_\infty)|\G^0_t]=p^i_0+\int_0^{T^0\wedge t}\frac{dp^i_s}{1-a_s}.\]
The integrals above are well defined since $\P(a^0_t=1)=0$ fot $t$ in the stochastic interval $[0,T^0]$. Also, since by construction $\sum p^i_t=1$, $\forall t$ it follows that $\sum q^i_t=\sum p^i_0=1$. Finally, $q^i_t\geq 0$ and we have thus checked that the conditional law in (\ref{condlawS}) is well defined.

As in the proof of Theorem \ref{thmrand1} (using now the filtration $ \GG^0$ instead of $\FF$), let us denote $E^i=\{S=i\}\bigcap\{T^i\leq t\}$ hence
\begin{equation}\label{interm}
\P(\tau\leq t|\G^0_s)= \P\left(\bigcup_{i\geq 0}E^i|\G^0_s\right) =\sum_{i\geq 1}\I_{\{T^i\leq t\}} q^i_s.
\end{equation} 
Notice that since $\FF \hookrightarrow \GG^0$, we can use Proposition \ref{bremyorproj} (i) and the fact that from (\ref{a}) $a_t=^o(\I_{T^0\leq t})$ and is continuous in oder to get:
\[
\;^o(q^i)=p^i_0+\;^o\left(\int_0^ \cdot \frac{\I_{T^0\geq s}}{1-a_s}dp^i_s\right)=p^i_0+\int_0^ \cdot \frac{\;^o(\I_{T^0\geq \cdot})_s}{1-a_s}dp^i_s=p^i.
\]
We denote $\tilde q_t=p^0_0+\int_0^{ t}\frac{dp^0_s}{1-a_s}$. Therefore, projecting onto the filtration $\FF$ the equality (\ref{interm}), leads us  to:
\begin{align*}
\P(\tau\leq t|\F_t)&=\sum_{i\geq 1}\;^o(q^i)_{t} \I_{\{T^i\leq t\}} +\;^o\left(\tilde q_{T^0}\I_{\{T^0\leq \cdot\}}\right)_t = \sum_{i\geq 1}p^i_t\I_{\{T^i\leq t\}}+\;^o\left(\int_0^\cdot  \tilde q_{s}d\I_{\{T^0\leq s\}}\right)_t\\
&=\sum_{i\geq 1}\Delta A_{T^i}\I_{\{T^i\leq t\}}+\int_0^t  \tilde q_{s}da_s.
\end{align*}
It is remains to check that $A^c_t=\int_0^t  \tilde q_{s}da_s$. Let us denote $\alpha_t= \int_0^t  \tilde q_{s}da_s$. After integration by parts one gets that   $\alpha_t= p^0_t-   \tilde q_{t}(1-a_t)$. Using this  and the expression of the process $a$ we have that $\alpha$ solves:
\[
d\alpha_t=\frac{p^0_t-\alpha_t}{p^0_t-A^c_t}dA^c_t
\]with $\alpha_0=A^c_0=0$ which has as solution $\alpha=A^c$.
\finproof

\begin{rem}

\begin{itemize}
\item[(i)] Suppose that the process $(A_t)$ given in the above theorem is continuous. Then, $T^i=\infty$ for $i\geq 1$ and $p^0_t=1$. It follows that $a\equiv A$, hence we obtain the classical construction we have presented in (\ref{constr1}).
\item[(ii)] The time $\tau$ constructed above is totally inaccessible if and only if the jump times of $A$ are totally inaccessible.
\end{itemize}
\end{rem}

\section{A decomposition of the default event risk premium} \label{sec::premium}

We recall our probability space is  $(\Omega, \F,\FF,\P)$ and $\tau$ is a random time which is not an $\FF$ stopping time and $\GG$ is the progressively enlarged filtration which makes $\tau$ a stopping time. We also continue to assume the following:
\begin{itemize}
\item[\textbf{ (H)}] $\FF\hookrightarrow \GG$, that is the filtration $\FF$ is immersed in $\GG$;
\item[\textbf{(TI)}] $\tau$ is a $\GG$  stopping time and $(T^i)_{i\in \mathbb{N}}$ is a decomposing sequence of $\tau$, i.e.,
\begin{equation*}
\tau=\sum_{i\geq 0} T^i \I_{\{T^i=\tau\}}.
\end{equation*}where $(T^i)_{i\geq 1}$ are $\FF$ stopping times, such that $\P(T^i=T^j <\infty) =0, \quad i\neq j$ and  $T^0$   avoids all finite $\FF$  stopping times.
\end{itemize}

The aim of this section is to provide an analysis of the risk premiums that  form in a financial market in anticipation of  losses which occur at the default event $\tau$. From this perspective, it can be added in the assumption \textbf{(TI)} that the time $\tau$ is totally inaccessible (this assumption will be added later on).  The defaultable claims  we are going to analyze are $\G_T$ measurable random variables ($T>0$ constant) that  have the specific form:
\begin{equation}\label{defpayoff2}
X=P\I_{\{\tau> T\}}+C_\tau\I_{\{\tau\leq T\}},
\end{equation} where we assume that $P$ is a positive square integrable,  $\F_T$-measurable random variable which represents a single payment which occurs at time $T$ and $(C_t)$ is a positive bounded,  $\FF$-adapted process.  $P$ stands for the promised payment, while the process $C$ models the recovery in case of default.

We do not assume the recovery process $C$ to be predictable, as it is common  in the usual reduced-form setting, since we want to emphasize possible drops in the values of the collateral when  macro-economic shocks arrive, which we believe are important phenomena in the credit markets, and should bear risk premiums.
In order to obtain explicit formulas for the risk premiums attached to the jumps of the price process at the macro-economic shock times $(T^i)_{i\geq 1}$, the following projection result will play an important role:
\begin{lem}[\cite{LeJan}]\label{lejan}
Let $T$ be an $\FF$ stopping time and $X$ a square integrable random variable which is $\F_T$ measurable. Then, there exists an $\FF$ predictable process $(x_t)$ such that $\E[X|\F_{T-}]=x_T$. Denote  $\xi=X-x_T$. The process $(\xi\I_{\{T\leq t\}},t\geq 0)$ is a square integrable martingale which is orthogonal to any square integrable martingale $M$ which has the property that $M_T$ is $\F_{T-}$ measurable.
\end{lem}

We therefore assume without loss of generality the following form for the recovery process:
\begin{itemize}
\item[\textbf{(R)}] There exist a sequence of predictable processes $(c^i)_{t\geq 0}$, $i\in \mathbb{N}^*$ and a sequence of random variables $(\kappa^i)_{i\in \mathbb{N}^*}$ such that  the recovery decomposes as:
\[C_t=\hat C_t-\sum_{i\geq 1}(c^i_{T^i}+\kappa^i)\I_{\{T^i\leq t\}},
\] where the process $(\hat C_t)_{t\geq 0}$ does not have discontinuities at the stopping times $T^i, i\in \mathbb{N}^*$ and the processes $(\kappa^i\I_{\{T^i\leq t\}})$ are martingales.
\end{itemize}

Also, the martingales $p^i$, $ i\geq 1$ admit the decomposition:
\begin{equation}\label{reprpi}
p^i_t=\int_0^t \upsilon^i_sdN^i_s + \phi^i\I_{\{T^i\leq t\}} + \hat m^i_t,
\end{equation}
where the processes $\upsilon^i$ are predictable with $\E[|\upsilon^i_{T^i}|]<\infty$; $\phi^i$ are random variables $\F_{T^i}$ measurable with $\E[\theta^i|\F_{T^i-}]=0$ and $(\hat m^i_t)$ is a martingale orthogonal to $N^i$.

Let us also denote $R_t=\int_0^tr_udu$, where $(r_t)$ is the locally risk-fee interest rate. We shall assume that $\P$ is a risk neutral measure. We recall that an arbitrage-free price of a defaultable claim is given by the following conditional expectation:
\[
S(X)_t:=e^{R_t}\E[Pe^{-R_T}\I_{\{\tau> T\}}+C_\tau e^{-R_\tau}\I_{\{\tau\leq T\}}|\G_t].
\]

Using the enlargement of filtration framework, pre-default prices can always be expressed in terms of an $\FF$-adapted process, via projections on the smaller filtration $\FF$ as (see \cite{elliotjeanbyor}, \cite{jenbrutk1}, \cite{bsjeanblanc}):
\[
\I_{\{\tau>t\}}S(X)_t=\I_{\{\tau>t\}}\tilde S(X)_t
\]
where $\tilde S(X)$ is $\FF$-adapted, given by:
\begin{equation}\label{predefprice}
\tilde S(X)_t:=\frac{e^{R_t}}{Z^\tau_t}\E[Pe^{-R_T}\I_{\{\tau> T\}}+C_\tau e^{-R_\tau}\I_{\{\tau\in (t, T]\}}|\F_t].
\end{equation}
The process $\tilde S(X)$ is always well defined on the stochastic interval $[0,\tau)$. We add the assumption that $Z^\tau>t, \forall t$ so that the pre-default price process is well defined $\forall t$. Therefore, we exclude the situation where $\tau$ is an $\FF$ stopping time.

We recall below a well known expression of the pre-default price process (see for instance\cite{elliotjeanbyor}, \cite{bsjeanblanc}, \cite{jenbrutk1}) which holds in a particular case of our framework: 
\begin{prop}[\cite{bsjeanblanc}] \label{propbsjeanblanc}Suppose that the process $C$ is predictable and the process $Z^\tau$ is continuous.  Then
 \begin{equation}\label{eqSimple2}
\tilde S(X)_t=\frac{e^{R_t}}{Z^\tau_t}\E\left[\int_t^TC_u e^{-(R_u+\Lambda_u)}d\Lambda_u+Pe^{-(R_T+\Lambda_T)}|\F_t\right],\quad t\geq 0.
 \end{equation}
\end{prop}
\begin{rem} It is known that when $\FF\hookrightarrow \GG$, $Z^\tau$ is continuous if and only if $\tau$ avoids the $\FF$ stopping times, i.e.,  in the condition \textbf{(TI)} all $T^i=\infty$ for $i\geq 1$.  
\end{rem}

We are now going to generalize the expression of the pre-default price  to our setting, synthesized in the assumptions \textbf{(H)}, \textbf{(TI)} and \textbf{(R)} and deduce an expression of the default event risk premium. We shall use the general pricing methodology developed in \cite{coculescunikeghbali}. 

To begin, we clarify what we understand by default event risk premium, by introducing the following definition. 

 \begin{defn}\label{def::premiums} Suppose that the pre-default price process $\tilde S(X)$ of the claim $X$ introduced in (\ref{defpayoff2}) has a Doob Meyer decomposition under the risk neutral measure $\P$:
 \begin{equation}\label{premium}
 \tilde S(X)_t= \tilde S(X)_0+ \int_0^t  \tilde S(X)_u d\nu(X)_u +M_t 
 \end{equation}
 where $(\nu(X)_t), t\geq 0$ is a finite variation, predictable process and $(M_t)$ a martingale $M_0=0$.  We call (cumulated) default risk premium the process $\pi(X)_t=\nu(X)_t-R_t$, $0\leq t\leq T$.

 \end{defn}
Intuitively, the default event risk premium represents  the additional net yield an investor can earn from a security as a compensation for the losses arriving at the default time. Indeed,  $\int_0^{\cdot\wedge\tau}   \tilde S(X)_u d\pi(X)_u$ represents the compensator of the jump (in practice a loss) that will occur in the price of the claim $X$ at the default time $\tau$. Hence, for totally inaccessible default times, the (cumulated) default event risk premium is an increasing, continuous process while for a predictable default time, this is always null.
\begin{ex} Suppose the process $Z^\tau$ is continuous and $C$ is predictable as in Proposition \ref{propbsjeanblanc}. Then from equation (\ref{eqSimple2}), it can be easily checked that:
\[
\pi(X)_t=\Lambda_t-\int_0^t \tilde C_u d\Lambda_u
\]
where $\tilde C= C/\tilde S(X)$. 
\end{ex}
Definition \ref{def::premiums} can be put in relation with  the notion of the instantaneous credit spread, i.e.,  the artificial discounting rate that one would need to apply to the promised payment, in excess to the risk-free rate, in order to make the price of the defaultable claim equal the expected discounted promised payment. The instantaneous credit spread (when it exists)   is an $\FF$ adapted process $(s_t)$ that satisfies the equality:
\[
\tilde S(X)_t=\E[Pe^{-\int_t^T (r_u+s_u)du)}|\F_t].
\]
After noticing that $\tilde S(X)_T=P$, we find that $S(X)$ decomposes as in  (\ref{premium}) with $\pi(X)_t=\int_0^t s_udu$. Therefore, the instantaneous credit spread exists only if the default risk premium $\pi(X)$ is absolutely continuous to Lebesgue measure, and in this case $\pi(X)$ can be interpreted as  the cumulated instantaneous spreads.

For pricing the defaultable claims, the following martingale will play a crucial role (as developed in \cite{coculescunikeghbali}):

\begin{defn}
We introduce the following exponential local martingale: 
\[D_t:=\mathcal E\left(\int_0^\cdot \frac{dm^\tau_s}{Z^\tau_{s-}}\right)_t=\mathcal E\left(-\int_0^\cdot \frac{d\hat N_s}{Z^\tau_{s-}}\right)_t,
\] 
where the martingale $\hat N$ was defined in the equality (\ref{Nhat}).
If $(D_t)_{0\leq t\leq T}$ is a square integrable martingale, we define the default-adjusted measure as:
\[
d\Q^\tau:=D_T \cdot d\P \quad \text{ on } \F_T.
\]
\end{defn}

\begin{rem}
It is useful to have in mind another equivalent expression of the local martingale $D$, namely: $$D_t=Z^\tau_te^{\Lambda_t}.$$
\end{rem}

We thus obtain the following  proposition:

\begin{prop}\label{price} Suppose that $T^i$, $i\geq 1$ are totally inaccessible stopping times. Furthermore, assume that  $(D_t)_{0\leq t\leq T}$ is a square integrable martingale (for instance $\E[e^{2\Lambda_T}]<\infty$). Then, the  pre-default price of the defaultable claims is given by:
\begin{equation}\label{pricegen}
 \tilde S_t(X)=e^{\tilde R_t}\E^{\Q^\tau}\left[\int_t^T e^{-\tilde R_{u}}C_{u-} d\Lambda_u -\sum_{i\geq 1}\int_t^T  e^{-\tilde R_{u}} h^i_u\frac{p^i_{u-}+\upsilon^i_u}{Z^\tau_{u-}}d\Lambda^i_{u\wedge T^i}  + Pe^{-\tilde R_T}|\F_t\right]\quad t<T,
 \end{equation}
 where: 
 \begin{align*}
 \tilde R_t& :=R_t+\Lambda_t,\\
 h^i_t&:=c^i_t+\frac{\tilde \kappa^i_t}{p^i_{t-}+\upsilon^i_t}
 \end{align*}
  and where, for $i\geq 1$, $\tilde \kappa^i$ is the predictable process which satisfies $\E[\kappa^i\phi^i|\F_{T^i-}]=\tilde \kappa^i_{T^i}$.

 In particular, if the recovery process $C$ does not have discontinuities at the stopping times $T^i, i\geq 1$, then the pre-default price of the defaultable claims is given by:
\begin{equation}\label{priceCpred}
 \tilde S_t(X)=e^{\tilde R_t}\E^{\Q^\tau}\left[\int_t^TC_{u} e^{-\tilde R_{u}}d\Lambda_u   + Pe^{-\tilde R_T}|\F_t\right]\quad t<T.
 \end{equation}

\end{prop}

\proof
The pricing formulas can be derived using arguments similar to those in \cite{coculescunikeghbali} (Proposition  4.3.), where $C$ was supposed predictable and $\Lambda$ continuous. We shall extend the formulas to hold for more general $C$ and $\Lambda$  (however in that paper, the \textbf{(H)} hypothesis was not necessarily holding).

First, let us notice from (\ref{predefprice})  that:
\[
\tilde S(P)_t=\frac{e^{R_t}}{Z^\tau_t}\E\left[Pe^{- R_T}Z^\tau_T|\F_t\right]=\frac{e^{\tilde R_t}}{D_t}\E\left[Pe^{-\tilde R_T}D_T|\F_t\right]=e^{\tilde R_t}\E^{\Q^\tau}\left[Pe^{-\tilde R_T}|\F_t\right].
\] 

By linearity of the conditional expectation, we only need to show the equality:
\[
\tilde S(C_\tau\I_{\{\tau \leq T\}})_t=e^{\tilde R_t}\E^{\Q^\tau}\left[\int_t^T e^{-\tilde R_u}C_{u-} d\Lambda_u -\sum_{i\geq 1}\int_t^T  e^{-\tilde R_u} h^i_u\frac{p^i_{u-}+\upsilon^i_u}{Z^\tau_{u-}}d\Lambda^i_{u\wedge T^i} |\F_t\right].
\]
Indeed, using equation (\ref{predefprice}), the pre-default price of a claim that pays $C_\tau$ at default and zero otherwise is given by:
\begin{align*}
\tilde S(C_\tau\I_{\{\tau \leq T\}})_t&=\frac{e^{R_t}}{Z^\tau_t}\E\left(e^{-R_\tau}C_\tau\I_{\{\tau\in(t, T]\}}|\F_t\right)=e^{\tilde R_t}D^{-1}_t\E\left(\int_t^Te^{-R_u}C_udA^\tau_u|\F_t\right)\\
&=e^{\tilde R_t}D^{-1}_t\sum_{i\geq 1}\E\left( e^{-R_{T^i}}C_{T^i}p^i_{T^i}\I_{\{T^i\in(t,T]\}}|\F_t\right)\\
&=e^{\tilde R_t}D^{-1}_t\sum_{i\geq 1}\E\left( e^{-R_{T^i}}\{C_{T^i-}-(c^i_{T^i}+\kappa^i)\}(p^i_{T^i-}+\upsilon^i_{T^i}+\phi^i)\I_{\{T^i\in(t,T]\}}|\F_t\right).
\end{align*}
Notice that (from Lemma \ref{lejan}) the following processes: $(e^{-R_{T^i}}(C_{T^i-}-c^i_{T^i})\phi^i\I_{\{T^i\leq t\}})$ and $(e^{-R_{T^i}}\kappa^i(p^i_{T^i-}+\upsilon^i_{T^i})\I_{\{T^i\in(t,T]\}})$  are martingales. Therefore we obtain:
\begin{align*}
\tilde S(C_\tau\I_{\{\tau \leq T\}})_t&=e^{\tilde R_t}D^{-1}_t\sum_{i\geq 1}\E\left(\int_t^T e^{-R_u}(C_{u-}-h^i_{u})(p^i_{u-}+\upsilon^i_{u})d\Lambda^i_{u\wedge T^i}|\F_t\right)\\
&=e^{\tilde R_t}D^{-1}_t\sum_{i\geq 1}\E\left(\int_t^T e^{-\tilde R_{u}}D_{u}(C_{u-}-h^i_{u})\frac{(p^i_{u-}+\upsilon^i_{u})}{Z^\tau_{u-}}d\Lambda^i_{u\wedge T^i}|\F_t\right)
\end{align*}

Denote $H^i_t:=\int_0^t e^{-\tilde R_{u}}(C_{u-}-h^i_{u})\frac{(p^i_{u-}+\upsilon^i_{u})}{Z^\tau_{u-}}d\Lambda^i_{u\wedge T^i}$. An integration by parts of the products $H^iD$ gives (recall that $\Lambda^i$ are supposed continuous):
\[
D_TH^i_T=D_tH^i_t+\int_t^T D_{u-}dH^i_u+\int_t^T H^i_{u} dD_u.
\]
Since $C$, $h^i$ and $\frac{(p^i_{t-}+\upsilon^i_t)}{Z^\tau_{t-}} $ are bounded and $D$ is a square integrable martingale, $H^i$ is also bounded and hence $\int H^idD$ is also a square integrable martingale. Consequently:
\[
\E\left[\int_t^T D_{u-}dH^i_u |\F_t\right]=\E\left[D_TH^i_T-D_t H^i_t-\int_t^T H^i_udD_u |\F_t\right]=\E\left[D_T(H^i_T-H^i_t)|\F_t\right].
\]
Therefore the pricing formula simplifies to :
\[
\tilde S(C_\tau\I_{\{\tau \leq T\}})_t=-e^{\tilde R_t}D^{-1}_t\E \left(D_T \sum_{i\geq 1}\int_t^T dH^i_u|\F_t\right),
\]
which, after a Girsanov transformation, leads to the desired formula.

\finproof

 As a corollary we obtain the following evolution of the pre-default price process (notice that the term $\tilde h^i$ is a risk premium attached to the sensitivity of the recovery to the economic shock $T^i$):
 
\begin{cor} Under the assumptions in Proposition \ref{price}, there exists a martingale $(M_t)$ under the measure $\P$ such that the pre-default price process satisfies the following stochastic differential equation:
\begin{equation}\label{sev}
\frac{d \tilde S_t(X)_t}{\tilde S_t(X)_t}=d\tilde R_t - \tilde C_t d\Lambda_t +\sum_{i\geq 1} \tilde h^i_t\frac{(p^i_{t-}+\upsilon^i_t)}{Z^\tau_{t-}}d\Lambda^i_{t\wedge T^i} +\sum_{i\geq 1}d\left (\frac{p^i_{T^i}\Delta M_{T^i}}{Z^\tau_{T^i-}}\I_{\{T^i\leq \cdot\}}\right)^p_t+dM_t
\end{equation}
where $\tilde C_t=C_t/\tilde S_t(X)_t$ and  $\tilde h^i_t=h^i_t/\tilde S_t(X)_t$.
\end{cor}

\proof

From equation (\ref{pricegen}) it follows that there exists a $\Q^\tau$-martingale $\tilde M$ such that:
\[
\frac{d \tilde S_t(X)_t}{\tilde S_t(X)_t}=d\tilde R_t - \tilde C_t d\Lambda_t +\sum_{i\geq 1} \tilde h^i_t\frac{p^i_{t-}}{Z^\tau_{t-}}d\Lambda^i_{t\wedge T^i}+d\tilde M_t
\] 
Using the Girsanov's theorem $\tilde M=M+\left (\sum_{i\leq 1}\frac{\Delta M_{T^i}p^i_{T^i}}{Z^\tau_{T^i-}}\I_{\{T^i\leq t\}}\right )^p$ where $M$ is a $\Q$ martingale. 
\finproof

In practical applications, one can obtain explicit formulas for the default event risk premium by computing explicitly the dual predictable projections appearing in the corollary above. Let us give the general formulation of these expressions. Any square integrable  $\FF$ martingale $M$ can be decomposed in sum of orthogonal martingales as follows (see \cite{LeJan}):
\begin{equation}\label{reprM}
M_t=\sum_{i\leq 1} \int_0^t f^i_sdN^i_s + \sum_{i\leq 1}\theta^i\I_{\{T^i\leq t\}} + \hat M_t,
\end{equation}
where the processes $f^i$ are predictable with $\E[|f^i_{T^i}|]<\infty$; $\theta^i$ are random variables $\F_{T^i}$ measurable with $\E[\theta^i|\F_{T^i-}]=0$ (i.e.,  the process $(\theta^i\I_{\{T^i\leq t\}})$ is a martingale) and  $(\hat M_t)$ is a martingale orthogonal to any  $N^i$, $i\geq 1$. 


\begin{thm}\label{decprem}
Suppose  that the martingale $M$ appearing in (\ref{sev}) decomposes as  in (\ref{reprM}) and the martingales $p^i$, $i\geq 1$ decompose as in (\ref{reprpi}).

Then,  the default event risk premium has the expression:
\begin{equation}\label{prime1}
\pi(X)_t=\int_0^t(1-\tilde C_u)d\Lambda_u+\sum_{i\geq 1}\int_0^t(\tilde h^i_u+\varphi^i_u) \frac{(p^i_{u-}+\upsilon^i_u)}{Z^\tau_{u-}}d\Lambda^i_{u\wedge T^i},
\end{equation}
where, for $i\geq 1$: $\varphi^i_t:=f^i_t+\frac{\sigma^i_t}{p^i_{t-}+\upsilon^i_t}$ and $(\sigma^i_t)$ is the predictable process which satisfies $$\E[\theta^i\phi^i|\F_{T^i-}]=\sigma^i_{T^i}.$$
\end{thm}

\proof

If $M$ has the representation (\ref{reprM}) then $\Delta M_{T^i} = f^i_{T^i}+\theta^i$.  Let us denote $n^i_t=\theta^i\I_{\{T^i\leq t\}} $ which we recall is a square integrable martingale orthogonal to square integrable martingales of the form $\int f_udN^i_u$. Therefore:
\begin{align*}
 \left(\frac{\Delta M_{T^i}p^i_{T^i}}{Z^\tau_{T^i-}}\I_{T^i\leq \cdot}\right)^p&=\left(\frac{ f^i_{T^i}p^i_{T^i}}{Z^\tau_{T^i-}}\I_{T^i\leq t}\right)^p+ \left(\frac{p^i_{T^i}}{Z^\tau_{T^i-}}\theta^i\I_{T^i\leq t}\right)^p\\
&=\int_0^t\frac{f^i_u (p^i_{u-}+\upsilon^i_u)}{Z^\tau_{u-}}d\Lambda^i_{u\wedge T^i}+ \left(\int_0^t \frac{p^i_u}{Z^\tau_{u-}}dn^i_u\right)^p
\end{align*}
We have that $\Delta p^i_{T^i}=\upsilon_{T^i}+\phi^i$ and hence: $\int_0^t p^i_udn^i_u=\int_0^t (p^i_{u-}+\upsilon^i_u)dn^i_u+\phi^i\theta^i_t\I_{\{T^i\leq t\}} $. The first term in the sum being a martingale, we obtain that  $ \left(\int_0^t p^i_udn^i_u\right)^p=  \left(\phi^i\theta^i_t\I_{\{T^i\leq t\}}\right)^p=\int_0^t \sigma^i_ud\Lambda_{u\wedge T^i}$ (as explained in \cite{delmey78}). The result follows.

\finproof

It is possible to decompose the default risk premium appearing in (\ref{prime1}) in an idiosyncratic and a systematic part (which in turn can be decomposed along premiums attached to each macroeconomic shock) as follows:
\[
\pi(X)_t=\int_0^t(1-\tilde C_u)\frac{da^0_u}{Z^\tau_{u-}}+\sum_{i\geq 1}\int_0^t\{1-\tilde C_u+(\tilde h^i_u+\varphi^i_u)\} \frac{p^i_{u-}+\upsilon_u}{Z^\tau_{u-}}d\Lambda^i_{u\wedge T^i}.
\]
We see that each possible macroeconomic shock $T^i$, $i\geq 1$ commands a corresponding risk premium for the possible jumps at $T^i$ of the recovery process but also of the martingale $M$, which reflects possible losses of a hedging portfolio.

\section{Conclusion}

In this paper we proposed a decomposition of a default times using sequences of stopping times of the reference filtration $\FF$. Our aim was to propose a systematic construction of default times but also to show that some of the simplifying assumptions appearing in the literature lead in fact to an underestimation of the risks attached to a long position in a defaultable claim, in particular the risks of losses at the default time. We hope that this analysis sheds light on the behavior of the assets at the default announcement (the so-called ''jump to default").

\begin{appendix}
\section {Some definitions}\label{sec::Appendix}

We consider a filtered probability space $(\Omega,\F,( \F_{t})_{t\geq0},\P)$ that satisfies the usual conditions.

\begin{thm}
[Optional and predictable projections]Let $X$ be a measurable process, positive or
bounded. There exists a unique (up to indistinguishability) optional
process $^{o}X$ (resp. predictable process $^{p}X$) such that:
\[
\mathbb{E}\left[  X_{T}\mathbf{1}_{\left(  T<\infty\right)  }|\mathcal{F}%
_{T}\right]  =\ ^{o}X_{T}\mathbf{1}_{\left(  T<\infty\right)  }\
a.s.
\]
for every stopping time $T$ (resp.
\[
\mathbb{E}\left[  X_{T}\mathbf{1}_{\left(  T<\infty\right)  }|\mathcal{F}%
_{T-}\right]  =\ ^{p}X_{T}\mathbf{1}_{\left(  T<\infty\right)  }\
a.s.
\]
for every predictable stopping time $T$). The process $^{o}X$ is called the
optional projection of $X$. The process $^{p}X$ is
called the predictable projection of $X$.
\end{thm}

A stochastic process which is nonnegative and whose path are
increasing and c\`adl\`ag, but which is not $\FF$-adapted is called a \textit{raw increasing
process}.

\begin{defn}
[Dual optional and predictable
projections]\label{DefDualPredProj1}Let $\left(  A_{t}\right)
_{t\geq0}$ be an integrable raw increasing process. We call dual
optional projection of $A$ the $\left(  \mathcal{F}_{t}\right)
$-optional increasing process $\left(  A_{t}^{o}\right)  _{t\geq0}$
defined
by:%
\[
\mathbb{E}\left[  \int_{0}^{\infty}\ ^{o}X_{s}dA_{s}\right]  =\mathbb{E}%
\left[  \int_{0}^{\infty}X_{s}dA_{s}^{o}\right]
\]
for any bounded adapted $\left(  X_{t}\right)  $. We call dual
predictable projection of $A$ the $\left(  \mathcal{F}_{t}\right)
$-predictable increasing process $\left(  A_{t}^{p}\right)
_{t\geq0}$, such that
\begin{equation}
\mathbb{E}\left[  \int_{0}^{\infty}\ ^{p}X_{s}dA_{s}\right]  =\mathbb{E}%
\left[  \int_{0}^{\infty}X_{s}dA_{s}^{p}\right]  .
\end{equation}
for every adapted bounded $\left(  X_{t}\right)  $.
\end{defn}

\end{appendix}

\renewcommand{\refname}{References}

\end{document}